\documentclass[a4paper,UKenglish,cleveref,autoref,thm-restate,authorcolumns]{lipics-v2019}

\usepackage{amsmath,amssymb}
\usepackage{bm,upgreek}
\usepackage[noend]{algpseudocode}
\usepackage{algorithm,algorithmicx}
\usepackage{pgf,tikz}
\usetikzlibrary{shapes,arrows,automata}
\usepackage{textcomp}
\usepackage{hyperref}
\usepackage{array}
\usepackage[braket,qm]{qcircuit}
\usepackage{dsfont}
\usepackage{lineno}
\usepackage{multicol}
\usepackage{graphics}

\newcommand{\QC}{\mathfrak{Q}}
\newcommand{\h}{\mathcal{H}}
\renewcommand{\L}{\mathcal{L}}
\newcommand{\D}{\mathcal{D}}
\newcommand{\I}{\mathcal{I}}
\newcommand{\II}{\mathbb{I}}

\newcommand{\M}{\mathds{M}}
\newcommand{\HH}{\mathbf{H}}
\newcommand{\LL}{\mathbf{L}}
\newcommand{\PP}{\mathbf{P}}
\newcommand{\T}{\mathrm{T}}
\newcommand{\dd}{\mathrm{d}}
\newcommand{\id}{\mathbf{I}}

\newcommand{\J}{\mathcal{J}}
\newcommand{\JJ}{\mathbb{J}}
\newcommand{\vl}{\mathrm{V2L}}
\newcommand{\lv}{\mathrm{L2V}}
\newcommand{\cq}{\mathrm{cq}}
\newcommand{\e}{\mathrm{e}}
\newcommand{\tr}{\mathrm{tr}}
\newcommand{\spn}{\mathrm{span}}

\newcommand{\diag}{\mathrm{diag}}

\newcommand{\x}{\mathbf{x}}
\newcommand{\ntl}{\mathrm{U}\,}
\newcommand{\mnt}{\mathrm{mnt}}
\newcommand{\true}{\texttt{true}}
\newcommand{\false}{\texttt{false}}

\renewcommand{\e}{\mathcal{E}}

\newtheorem{conjecture}[theorem]{Conjecture}

\algnewcommand\algorithmiccom{\textbf{Complexity:}}
\algnewcommand\Com{\item[\algorithmiccom]}

\title{Model Checking Quantum Continuous-Time Markov Chains}

\titlerunning{Model Checking Quantum CTMCs}

\author{Ming Xu}{Shanghai Key Lab of Trustworthy Computing,
East China Normal University, China}{mxu@cs.ecnu.edu.cn}{}{}
\author{Jingyi Mei}{Shanghai Key Lab of Trustworthy Computing,
East China Normal University, China}{mjyecnu@163.com}{}{}
\author{Ji Guan}{State Key Lab of Computer Science, Institute of Software,
Chinese Academy of Sciences, China}{guanji1992@gmail.com}{}{}
\author{Nengkun Yu}{Centre for Quantum Software and Information,
University of Technology Sydney, Australia}{nengkunyu@gmail.com}{}{}

\authorrunning{M.~Xu, J.~Mei, J.~Guan, and N.~Yu}

\Copyright{TBD}



\nolinenumbers 

\EventEditors{John Q. Open and Joan R. Access}
\EventNoEds{2}
\EventLongTitle{42nd Conference on Very Important Topics (CVIT 2016)}
\EventAcronym{CVIT}
\EventYear{2016}
\EventDate{December 24--27, 2016}
\EventLocation{Little Whinging, United Kingdom}
\EventLogo{}
\SeriesVolume{42}
\ArticleNo{23}

\begin{document}
	
\maketitle

\begin{abstract}
	Verifying quantum systems has attracted a lot of interests in the last decades.
	In this paper, we initialised the model checking of quantum continuous-time Markov chain (QCTMC).
	As a real-time system,
	we specify the temporal properties on QCTMC by signal temporal logic (STL).
	To effectively check the atomic propositions in STL,
	we develop a state-of-art real root isolation algorithm under Schanuel's conjecture;
	further, we check the general STL formula by interval operations with a bottom-up fashion,
	whose query complexity turns out to be linear in the size of the input formula
	by calling the real root isolation algorithm.
	A running example of an open quantum walk is provided to demonstrate our method.
\end{abstract}

\section{Introduction}
Aiming to study nature,
physicists use different mechanics depending on the scale of the objects
they are interested in.
Classical mechanics describes nature at macroscopic scale
(far larger than $10^{-9}$ meters),
while quantum mechanics is applied at microscopic scale (near or less than $10^{-9}$ meters). 
A particle at this level can be mathematically represented
by a normalised complex vector $\ket{s}$ in a Hilbert space $\h$. 
The time evolution of a single particle \emph{closed} system is described by
the Schr\"odinger equation 
\begin{equation}\label{eq:schrodinger}
	\frac{\dd\ket{s(t)}}{\dd t} = -\imath \HH\ket{s(t)}
\end{equation}
with some \emph{Hamiltonian} $\HH$ (a Hermitian matrix on $\h$),
where $\ket{s(t)}$ is the state of the system at time $t$. 
More practically,
an \emph{open} quantum system interacting with the surrounding environment need to be considered.
Suffering noises from the environment,
the state of the system cannot be completely known.
Thus a \emph{density operator} $\rho$ (positive semidefinite matrix with unit trace) on $\h$
is introduced to describe the uncertainty of the possible states: 
$\rho = \sum_{i\in I} p_i\op{s_i}{s_i}$,
where $\{(p_i,\ket{s_i})\}_{i\in I}$ is a mixed state
or an ensemble expressing that the quantum state is at $\ket{s_i}$ with probability $p_i$,
and $\bra{s_i}$ is the complex conjugate and transpose of $\ket{s_i}$.
In this case, the evolution is described by the Lindblad's master equation:
\begin{equation}\label{Lind}
	\frac{\dd \rho(t)}{\dd t}=\L(\rho(t))
\end{equation}
where $\rho(t)$ stands for the (possibly mixed) state of the system,
and $\L$ is a linear function of $\rho(t)$ (to be formally described in Subsection~\ref{S22}),
which is generally irreversible.

To reveal physical phenomenon,
physicists have intensively studied the properties of closed and open quantum systems
case by case in the last decades,
such as long-term behaviors~(e.g. \cite{GLT+10})
and stabilities~(e.g. \cite{CiT15}).
In recent years,
the computer science community has stepped into this field
and adopted model checking technique to study quantum systems~\cite{GNP06,GNP08}.
Specifically, the quantum systems can be simulated by some mathematical models
and a bulk of physical properties can be reformulated as formulas
in some temporal logic with atomic propositions of quantum interpretation.
In particular, a \emph{quantum discrete-time Markov chain} (QDTMC) $(\h,\e)$
has been introduced as quantum generalisations of classical discrete-time Markov chains (DTMC)
to model the evolution in Eq.~\eqref{Lind} in a single time unit:
\begin{equation}\label{eq:evolution}
	\rho(t+1)=\e(\rho(t))
\end{equation}
where $\e$ is a discretised quantum operation obtained from the Lindblad's master equation,
usually called a \emph{super-operator} in the field of quantum information and computation.
Several fundamental model checking-related problems
for QDTMCs have been studied in the literature,
including limiting states~\cite{Wol12,GFY18},
reachability~\cite{YFY+13,XHF21},
repeated reachability~\cite{FHT+17},
linear time properties~\cite{LiF15},
and persistence based on irreducible and periodic decomposition techniques~\cite{BaN11,GFY18}.
These techniques were equipped to solve real-world problems in several different areas.
For example, \cite{FYY13} proposed algorithms to
model checking quantum cryptographic protocols
against a quantum extension of probabilistic computation tree logic (PCTL);
and linear temporal logic (LTL) was adopted to
specify a bulk of properties of quantum many-body and statistical systems,
and corresponding model checking algorithm was developed~\cite{GFT+19}.
See~\cite{YiF19,YiF21} for the comprehensive review of this research line. 

However, to the best of our knowledge,
there is no work on model checking the quantum continuous-time system of Eq.~\eqref{Lind}.
In contrast, there are fruitful results in the classical counterpart,
which is usually modelled by a continuous-time Markov chain (CTMC). 

The seminal work on verifying CTMCs is Aziz et~al.'s paper~\cite{ASS+96,ASS+00}.
The authors introduced continuous stochastic logic (CSL) interpreted on CTMCs.
Roughly speaking, the syntax of CSL amounts to
that of PCTL plus the multiphase until formula
$\Phi_1 \ntl^{\I_1} \Phi_2 \cdots \ntl^{\I_K} \Phi_{K+1}$,
for some $K \ge 1$.
Because \cite{ASS+96} restricts probabilities in $\Pr_{>\texttt{c}}$
to $\texttt{c} \in \mathbb{Q}$,
they can show decidability of model checking for CSL
using number-theoretic analysis.
An approximate model checking algorithm for a reduced version of CSL
was provided by Baier et~al.~\cite{BKH99},
who restrict path formulas to binary until: $\Phi_1 \ntl^\I \Phi_2$.
Under this logic, they successfully applied efficient numerical techniques
for \emph{transient analysis}~\cite{BHH+03}
using \emph{uniformisation}~\cite{Ste94}.
The approximate algorithms have been extended
for multiphase until formulas
using \emph{stratification}~\cite{ZJN+11,ZJN+12}.
Xu~\textit{et~al.} considered the multi-phase until formulas
over the CTMC with rewards~\cite{XZJ+16}.
An integral-style algorithm was proposed to attack this problem,
whose effectiveness is ensured by number-theoretic results and algebraic methods.
Recently, continuous linear logic was introduced to specified on CTMCS,
whose decidability was established~\cite{GuY20}.
Most the above algorithms have been implemented in probabilistic model checkers,
like \textsl{PRISM}~\cite{KNP11}, \textsl{Storm}~\cite{DJK+17},
and \textsl{ePMC}~\cite{HLS+14}.

Unfortunately, these results from CTMCs cannot directly tackle the problem of
automatically verifying quantum continuous-time systems.
The main obstacle is that the state space of classical case is finite,
while the space of quantum states in even a finite-dimensional Hilbert space
is continuum (i.e., infinite).
In this paper, we cross the difficulty
by introducing quantum continuous-time Markov chains (QCTMC)
to model the evolution of quantum continuous-time systems in Eq.~\eqref{Lind}
and converting them into a distribution transformer
that preserves the laws of quantum mechanics.
Then, we consider a wide logic, signal temporal logic (STL),
to specify real-time properties against QCTMC.
The STL is more expressible than LTL and CTL.
Finally we present an exact method to decide the STL formula
using real root isolation and interval operations,
whose query complexity turns out to be linear in the size of the input formula
by calling the real root isolation routine.

The key contributions of the present paper are three-fold:
\begin{enumerate}
	\item In the field of formal verification,
		the model checking on DTMC, QDTMC and CTMC has been well studied in the past decades,
		but no work on checking QCTMCs as far as we know.
		The first contribution is filling this blank.
	\item In order to solve the atomic propositions in STL,
		we develop a state-of-the-art real root isolation algorithm
		for a rich class of real-valued functions based on Schanuel's conjecture.
	\item We provide a running example---%
		open quantum walk equipped with an absorbing boundary,
		which drops the restriction on the underlying graph being symmetric/Hermitian.
		The non-Hermitian structure brings real technical hardness.
		Fortunately, it is overcome by Eq.~\eqref{Lind}
		employed for describing the dynamical system of QCTMC.
\end{enumerate}

\subparagraph{Organisation}
The rest of the paper is structured as follows.
Section~\ref{S2} review some notions and notations
in quantum computing, the Lindblad's master equation and number theory.
In Sections~\ref{S3} and~\ref{S4},
we introduce the model of quantum continuous-time Markov chains
and the signal temporal logic (STL), respectively.
We solve the atomic propositions in STL in Section~\ref{S5},
and decide the general STL formulas in Section~\ref{S6}.
Section~\ref{S7} is the conclusion.

\section{Preliminaries}\label{S2}
Here we will recall some useful notions and notations from quantum computing~\cite{NiC00},
the Lindblad's master equation and number theory.

\subsection{Quantum Computing}
Let $\h$ be a Hilbert space with dimension $d$.
We employ the Dirac notations that are standard in quantum computing:
\begin{itemize}
	\item $\ket{s}$ stands for a unit column vector in $\h$ labelled with $s$;
	\item $\bra{s}:=\ket{s}^\dag$ is the Hermitian adjoint
	(complex conjugate and transpose) of $\ket{s}$;
	\item $\ip{s_1}{s_2}:=\bra{s_1}\ket{s_2}$
	is the inner product of $\ket{s_1}$ and $\ket{s_2}$;
	\item $\op{s_1}{s_2}:=\ket{s_1} \otimes \bra{s_2}$ is the outer product,
	where $\otimes$ denotes tensor product; and
	\item $\ket{s_1,s_2}:=\ket{s_1}\ket{s_2}$ is a shorthand of
    the product state $\ket{s_1}\otimes\ket{s_2}$.
\end{itemize}

A linear operator $\gamma$ is \emph{Hermitian} if $\gamma=\gamma^\dag$;
and it is \emph{positive}
if $\bra{s}\gamma\ket{s} \ge 0$ holds for any $\ket{s}\in\h$.
A \emph{projector} $\PP$ is a positive operator of
the form $\sum_{i=1}^m \op{s_i}{s_i}$ with $m\le d$,
where $\ket{s_i}$ are orthonormal.
Clearly, there is a bijective map
between projectors $\PP=\sum_{i=1}^m \op{s_i}{s_i}$
and subspaces of $\h$ that are spanned by $\{\ket{s_i}:1 \le i \le m\}$.
In sum, positive operators are Hermitian ones
whose eigenvalues are nonnegative;
and projectors are positive operators whose eigenvalues are $0$ or $1$.
Besides, a linear operator $\mathbf{U}$ is \emph{unitary}
if $\mathbf{U}\mathbf{U}^\dag=\mathbf{U}^\dag\mathbf{U}=\id$
where $\id$ is the \emph{identity} operator.

The \emph{trace} of a linear operator $\gamma$ is defined as
$\tr(\gamma):=\sum_{i=1}^d \bra{s_i}\gamma\ket{s_i}$
for any orthonormal basis $\{\ket{s_i}:1 \le i \le d\}$ of $\h$.
A \emph{density operator} $\rho$ on $\h$
is a positive operator with unit trace.
It gives rise to a generic way to describe quantum states:
if a density operator $\rho$ is $\op{s}{s}$ for some $\ket{s}\in \h$,
$\rho$ is said to be a \emph{pure} state;
otherwise it is a \emph{mixed} one,
i.e. $\rho=\sum_{i=1}^m p_i \op{s_i}{s_i}$ with $m\ge 2$ by spectral decomposition,
where $p_i$ are positive eigenvalues
(interpreted as the \emph{probabilities} of taking the pure states $\ket{s_i}$)
and their sum is $1$.
In other words, a pure state indicates the system state which we completely know;
a mixed state gives all possible system states, with total probability $1$, which we know.
We denote by $\D_\h$ the set of density operators on $\h$.
The subscript $\h$ of $\D_\h$ will be omitted if it is clear from the context. 

The system evolution between pure states is characterised by some unitary operator $\mathbf{U}$,
i.e. $\op{s(t)}{s(t)}=\mathbf{U}(t)\op{s(0)}{s(0)}\mathbf{U}^\dag(t)$
where $\mathbf{U}(t)$ comes from $\exp(-\imath \HH t)$
for the Hermitian operator $\HH$ in Eq.~\eqref{eq:schrodinger};
the system evolution between density operators (pure or mixed states) is characterised
by some completely positive operator-sum, a.k.a. \emph{super-operator},
i.e. $\rho(t)=\sum_{j=1}^m \LL_j(t) \rho(0) \LL_j^\dag(t)$
where $\LL_j$ are linear operators satisfying
the trace preservation $\sum_{j=1}^m \LL_j^\dag \LL_j=\id$.
The latter dynamical system is obtained in such a way:
an enlarged unitary operator acts on the purified composite state
$(\sum_{i=1}^m \sqrt{p_i}\ket{s_i,env_i})(\sum_{i=1}^m \sqrt{p_i}\bra{s_i,env_i})$
with $\rho(0)=\sum_{i=1}^m p_i \op{s_i}{s_i}$
for some orthonormal environment states $\ket{env_i}$~\cite[Section~2.5]{NiC00},
then the environment in the resulting composite state is traced out,
which turns out to be the aforementioned operator-sum form.
We call it an \emph{open} system as it interacts with the environment.
Whereas, the former dynamical system, independent from the environment,
is called a \emph{closed} system.
Open systems are more common than closed systems in practice.

\subsection{Lindblad's Master Equation}\label{S22}
To characterise state evolution of the continuous-time open system with the \emph{memoryless} property,
we employ the Lindblad's master equation~\cite{Lin76,GKS76} that is
\begin{equation}\label{eq:Lindblad}
	\rho'=\L(\rho)
	=-\imath\HH\rho+\imath\rho\HH
	+\sum_{j=1}^m \left(\LL_j\rho \LL_j^\dag
	-\tfrac{1}{2}\LL_j^\dag\LL_j\rho-\tfrac{1}{2}\rho\LL_j^\dag\LL_j\right),
\end{equation}
where $\HH$ is a Hermitian operator and $\LL_j$ are linear operators.
The terms $-\imath\HH\rho+\imath\rho\HH$ describe the evolution of the internal system;
the terms $\LL_j\rho \LL_j^\dag
-\tfrac{1}{2}\LL_j^\dag\LL_j\rho-\tfrac{1}{2}\rho\LL_j^\dag\LL_j$ 
describe the interaction between system and environment.
In other words,
to characterise the evolution of an open system,
it is necessary to use those linear operators $\LL_j$ besides the Hermitian operator $\HH$.
It is known to be the most general type of Markovian and time-homogeneous master equation
describing (in general non-unitary) evolution of the system state
that preserves the laws of quantum mechanics
(i.e., completely positive and trace-preserving for any initial condition).

In the following, we will derive the solution of Eq.~\eqref{eq:Lindblad}.
We first define two useful functions:
\begin{itemize}
	\item $\lv(\gamma):=\sum_{i=1}^d \sum_{j=1}^d \bra{i}\gamma\ket{j} \ket{i,j}$
	that rearranges entries of the linear operator $\gamma$
	on the Hilbert space $\h$ with dimension $d$
	as a column vector; and
	\item $\vl(\mathbf{v}):=\sum_{i=1}^d \sum_{j=1}^d \bra{i,j} \mathbf{v} \op{i}{j}$
	that rearranges entries of the column vector $\mathbf{v}$ as a linear operator.
\end{itemize}
Here, $\lv$ and $\vl$ are read as
``linear operator to vector'' and ``vector to linear operator'', respectively.
They are mutually inverse functions,
so that
if a linear operator (resp.~its vectorisation) is determined,
its vectorisation (resp.~the original linear operator) is determined.
Using the fact that $\mathbf{D}=\mathbf{A}\mathbf{B}\mathbf{C} \Longleftrightarrow
\lv(\mathbf{D})=(\mathbf{A} \otimes \mathbf{C}^\T)\lv(\mathbf{B})$
holds for any linear operators $\mathbf{A},\mathbf{B},\mathbf{C},\mathbf{D}$,
we can reformulate Eq.~\eqref{eq:Lindblad} as the linear ordinary differential equation
\begin{equation}\label{eq:ODE}
	\begin{aligned}
		\lv(\rho') &=\left[-\imath\HH\otimes\id+\imath\id\otimes\HH^\T
		+\sum_{j=1}^m \left(\LL_j\otimes\LL_j^{*}
		-\tfrac{1}{2}\LL_j^\dag\LL_j\otimes\id-\tfrac{1}{2}\id\otimes \LL_j^\T\LL_j^*\right)\right]
		\lv(\rho) \\
		&= \M \cdot\lv(\rho),
	\end{aligned}
\end{equation}
where $*$ denotes entry-wise complex conjugate.
We call $\M=-\imath\HH\otimes\id+\imath\id\otimes\HH^\T
+\sum_{j=1}^m \big( \LL_j\otimes\LL_j^{*}
-\tfrac{1}{2}\LL_j^\dag\LL_j\otimes\id \linebreak[0]
-\tfrac{1}{2}\id\otimes \LL_j^\T\LL_j^* \big)$
the \emph{governing matrix} of Eq.~\eqref{eq:ODE}.
As a result,
we get the desired solution $\lv(\rho(t))=\exp(\M\cdot t)\cdot\lv(\rho(0))$
or equivalently $\rho(t)=\vl(\exp(\M\cdot t)\cdot\lv(\rho(0)))$ in a closed form.
It can be obtained in polynomial time by the standard method~\cite{Kai80}.

\subsection{Number Theory}
\begin{definition}
	A number $\alpha$ is \emph{algebraic},
	denoted by $\alpha \in \mathbb{A}$,
	if there is a nonzero $\mathbb{Q}$-polynomial $f_\alpha(z)$ of least degree,
	satisfying $f_\alpha(\alpha)=0$;
	otherwise $\alpha$ is \emph{transcendental}.
\end{definition}
In the above definition, such a polynomial $f_\alpha(z)$ is called
the \emph{minimal polynomial} of $\alpha$.
The \emph{degree} $D$ of $\alpha$ is $\deg_z(f_\alpha)$.
The standard encoding of $\alpha$ is the minimal polynomial $f_\alpha$
plus an isolation disk in the complex plane
that distinguishes $\alpha$ from other roots of $f_\alpha$.

\begin{definition}
Let $\mu_1,\ldots,\mu_m$ be irrational complex numbers.
Then the \emph{field extension} $\mathbb{Q}(\mu_1,\ldots,\mu_m):\mathbb{Q}$
is the smallest set
that contains $\mu_1,\ldots,\mu_m$ and is closed under arithmetic operations,
i.e. addition, substraction, multiplication and division.
\end{definition}
Here those irrational complex numbers $\mu_1,\ldots,\mu_m$ are called
the generators of the field extension.
A field extension is \emph{simple} if it has only one generator.
For instance, the field extension $\mathbb{Q}(\sqrt{2}):\mathbb{Q}$
is exactly the set $\{a+b\sqrt{2} : a,b \in \mathbb{Q}\}$.

\begin{lemma}[{\cite[Algorithm~2]{Loo83}}]\label{lem:simple}
	Let $\alpha_1$ and $\alpha_2$ be two algebraic numbers of
	degrees $D_1$ and $D_2$, respectively.
	There is an algebraic number $\mu$ of degree at most $D_1 D_2$,
	such that the field extension $\mathbb{Q}(\mu):\mathbb{Q}$
	is exactly $\mathbb{Q}(\alpha_1,\alpha_2):\mathbb{Q}$.
\end{lemma}
For a collection of algebraic numbers $\alpha_1,\ldots,\alpha_m$
appearing in the input instance,
by repeatedly applying this lemma,
we can obtain a simple field extension $\mathbb{Q}(\mu):\mathbb{Q}$
that can span all $\alpha_1,\ldots,\alpha_m$.
\begin{lemma}[{\cite[Corollary~4.1.5]{Coh96}}]\label{lem:closed}
	Let $\alpha$ be an algebraic number of degree $D$,
	and $g(z)$ an $\mathbb{A}$-polynomial with degree $D_g$
	and coefficients taken from $\mathbb{Q}(\alpha):\mathbb{Q}$.
	There is a $\mathbb{Q}$-polynomial $f(z)$ of degree at most $DD_g$,
	such that roots of $g(z)$ are those of $f(z)$.
\end{lemma}
The above lemma entails that
roots of any $\mathbb{A}$-polynomial are also algebraic.

\begin{theorem}[Lindemann (1882)~{\cite[Theorem~1.4]{Bak75}}]\label{Lindemann}
	For any nonzero algebraic numbers $\beta_1,\ldots,\beta_m$ and
	any distinct algebraic numbers $\lambda_1,\ldots,\lambda_m$,
	the sum $\sum_{i=1}^m \beta_i \mathrm{e}^{\lambda_i}$ with $m \ge 1$ is nonzero.
\end{theorem}

\begin{conjecture}[Schanuel (1960s)~\cite{Ax71}]\label{Schanuel}
	Let $\lambda_1,\ldots,\lambda_m$ be
	$\mathbb{Q}$-linearly independent complex numbers.
	Then the field extension
	$\mathbb{Q}(\lambda_1,\mathrm{e}^{\lambda_1},\ldots,\lambda_m,\mathrm{e}^{\lambda_m}):\mathbb{Q}$
	has transcendence degree at least $m$.
\end{conjecture}

Let $\mathbb{A}[z_1,\ldots,z_m]$ denote the ring
that contains all $\mathbb{A}$-polynomials in variables $z_1,\ldots,z_m$.
Assuming Schanuel's conjecture, we could get:
\begin{corollary}[{\cite[Proposition~5]{COW16}}]\label{cor:common}
	Let $\lambda_1,\ldots,\lambda_m$ be
	$\mathbb{Q}$-linearly independent algebraic numbers.
	Then two co-prime elements $\varphi_1$ and $\varphi_2$
	in the ring $\mathbb{A}[t,\exp(\lambda_1 t),\ldots,\exp(\lambda_m t)]$
	have no common root except for $0$.
\end{corollary}

\section{Quantum Continuous-Time Markov Chain}\label{S3}
In this section, we propose the model of quantum continuous-time Markov chain (QCTMC).
We will reveal that it extends the classical continuous-time Markov chain (CTMC).
To show the practical usefulness,
an example is further provided for modelling open quantum walk.

For the sake of clarity,
we start with the QCTMC without classical states:

\begin{definition}\label{def:QCTMC1}
	A \emph{quantum continuous-time Markov chain} $\QC$ is a pair $(\h,\L)$,
	in which
	\begin{itemize}
		\item $\h$ is the Hilbert space,
		\item $\L$ is the transition generator function given by
		a Hermitian operator $\HH$
		and a finite set of linear operators $\LL_j$ on $\h$.
	\end{itemize}
	Usually, a density operator $\rho(0) \in \D$
	is appointed as the initial state of $\QC$.
\end{definition}

In the model,
the transition generator function $\L$ gives rise to a \emph{universal} way
to describe the bahavior of the QCTMC,
following the generality of the Lindblad's master equation.
Thus the state $\rho(t)$ is given by the closed-form solution to Eq.~\eqref{eq:ODE},
i.e., $\vl(\exp(\M\cdot t)\cdot\lv(\rho(0)))$, 
where $\M$ is the governing matrix for $\L$,
is a computable function from $\mathbb{R}_{\ge 0}$ to $\mathbb{C}^{d \times d}$.
We notice that $0 \le \tr(\PP \rho(t)) \le 1$ holds for any projector $\PP$ on $\h$,
as $\rho(t)$ is a density operator on $\D$.
Considering computability, the entries of $\HH$, $\LL_j$ and $\rho(0)$ are supposed to be algebraic.

Next, we equip the QCTMC in Definition~\ref{def:QCTMC1}
with finitely many classical states.
\begin{definition}\label{def:QCTMC2}
	A \emph{quantum continuous-time Markov chain} $\QC$	with a finite set $S$ of classical states
	is a pair $(\h_\cq,\L)$, in which
	\begin{itemize}
		\item $\h_\cq:=\mathcal{C} \otimes \h$ is the classical--quantum system
		with $\mathcal{C}=\spn(\{\ket{s}: s\in S\})$, and
		\item $\L$ is the transition generator function given by
		a Hermitian operator $\HH$
		and a finite set of linear operators $\LL_j$ on the enlarged $\h_\cq$.
	\end{itemize}
	Usually, a density operator $\rho(0) \in \D_{\h_\cq}$
	is appointed as the initial state of $\QC$.
\end{definition}

In fact, the models in Definitions~\ref{def:QCTMC1} and~\ref{def:QCTMC2}
have the same expressibility:
The QCTMC in Definition~\ref{def:QCTMC1} can be obtained
by setting the singleton state set $S=\{s\}$ of the QCTMC in Definition~\ref{def:QCTMC2};
conversely, the QCTMC in Definition~\ref{def:QCTMC2} can be obtained
by setting the Hilbert space as $\h_\cq$ of the QCTMC in Definition~\ref{def:QCTMC1}.
Hence, we can freely choose one of the two definitions for convenience.
As an immediate result, using Definition~\ref{def:QCTMC2},
we can easily see that the QCTMC extends the CTMC by the following lemma:
\begin{lemma}
	Given a CTMC $\mathfrak{C}=(S,\mathbf{Q})$,
	it can be modelled by a QCTMC $\QC=(\mathcal{C} \otimes \h,\L)$
	with $\mathcal{C}=\spn(\{\ket{s}:s\in S\})$ and $\dim(\h)=1$.
\end{lemma}
\begin{proof}
    It suffices to show that
    the states of a CTMC $\mathfrak{C}$ can be obtained by those of some QCTMC $\QC$.
	The state $\x=(x_s)_{s \in S}$ of $\mathfrak{C}$ is given by
	the dynamical system $\x'(t)=\x(t) \cdot \mathbf{Q}$
	or equivalently its closed-form solution
	$\x(t)=\x(0) \cdot \exp(\mathbf{Q}\cdot t)$,
	where $\x(0)$ is a row vector interpreted as the initial state.
	We construct the QCTMC by setting Hermitian operator $\HH=0$
	and linear operators $\LL_{s,t}=\op{t}{s}\otimes \mathbf{Q}[s,t]$ in $Q$
	for each pair $s,t \in S$.
	It is not hard to validate that the state $\rho(t)$ of $\QC$
	is $\diag(\x(t))$ of $\mathfrak{C}$,
	thus the lemma follows.
\end{proof}
	
\begin{example}[Open Quantum Walk~\cite{Pel14,SiP15}]\label{ex1}
	Open quantum walk (OQW) is a quantum analogy of random walk,
	whose system evolution interacts with environment.
	For the sake of clarity, we suppose that
	a particle walking along the $2$-dimensional hypercubic shown in Figure~\ref{fig:QW}.
	The position set is $S = \{s_{00}, s_{01}, s_{10}, s_{11}\}$,
	where $s_{00}$ denotes the starting position, a.k.a. the entrance,
	$s_{01}$ and $s_{10}$ denote transient positions,
	and $s_{11}$ denotes the exit, an absorbing boundary.
	The direction set is $\{\mathrm{F},\mathrm{S}\}$,
	where $\mathrm{F}$ means the particle takes the external transition along the first coordinate,
	while $\mathrm{S}$ means the particle takes the external transition along the second coordinate.
	The particle will choose a direction at every moment by the inner quantum ``coin-tossing''
	before being absorbed.
	This action is implemented by the Hadamard operator
	$H = \op{+}{\mathrm{F}} + \op{-}{\mathrm{S}}$
	with $\ket{\pm} = (\ket{\mathrm{F}}\pm\ket{\mathrm{S}})/\sqrt{2}$,
	which denotes a fair selection between $\mathrm{F}$ and $\mathrm{S}$
	as the probability amplitudes of both directions are $\tfrac{1}{2}=(\pm1/\sqrt{2})^2$.
	
	The OQW is modelled by a QCTMC $\QC_1=(\mathcal{C}\otimes\h,\L)$
	with $\mathcal{C}=\spn(\{\ket{s}:s\in S\})$,
	where each position in $S$ represents a classical state,
	and the transition function $\L$ is given by the Hermitian operator $\HH=0$
	and the unique linear operator
	\[
	\begin{aligned}
		\LL = & \op{s_{01}}{s_{00}}\otimes\op{\mathrm{S}}{-}
			+ \op{s_{10}}{s_{00}}\otimes\op{\mathrm{F}}{+}
			+ \op{s_{00}}{s_{01}}\otimes\op{\mathrm{S}}{-} \ + \\
			& \op{s_{11}}{s_{01}}\otimes\op{\mathrm{F}}{+}
			+ \op{s_{00}}{s_{10}}\otimes\op{\mathrm{F}}{+} 
			+ \op{s_{11}}{s_{10}}\otimes\op{\mathrm{S}}{-}.
	\end{aligned}
	\]
	For instance, when the particle is in the position $\ket{s_{00}}$,
	we first apply the quantum coin-tossing $H$ to the state in $\h$,
	then we get the result $\mathrm{F}$ or $\mathrm{S}$
	leading to the position $\ket{s_{10}}$ or $\ket{s_{01}}$.
	The composite operations are $(\op{\mathrm{F}}{\mathrm{F}}) H = \op{\mathrm{F}}{+}$
	and $(\op{\mathrm{S}}{\mathrm{S}}) H = \op{\mathrm{S}}{-}$,
	which makes up the first two terms in the above $\LL$.
	
	From $\LL$,
	we could get the governing matrix
	$\M = \LL\otimes\LL^* - \tfrac{1}{2}\LL^\dag\LL\otimes\id - \tfrac{1}{2}\id\otimes\LL^\T\LL^*$.
	Let $\rho(0) = \op{s_{00}}{s_{00}}\otimes\op{\mathrm{F}}{\mathrm{F}}$ be an initial state.
	Then the states $\rho(t)$ of the OQW could be computed
	as $\vl(\exp(\M\cdot t)\cdot\lv(\rho(0)))$ (we omit the detailed value due to space limit). \qed
		
		\begin{multicols}{2}
		\begin{itemize}
		\item[] 
		\begin{minipage}{\linewidth}
		\captionsetup{type=figure}
		\begin{tikzpicture}[->,>=stealth',auto,node distance=2.5cm,semithick,inner sep=2pt]
			\node[state](s0){$s_{01}$};
			\node[state,initial,initial text={entrance}](s1)[below of=s0]{$s_{00}$};
			\node[state](s2)[right of=s1]{$s_{10}$};
		\node[state,accepting,label={right:exit}](s3)[right of=s0]{$s_{11}$};
			\draw(s1)edge[bend left]node[]{S}(s0);
			\draw(s0)edge[bend left]node[]{S}(s1);
			\draw(s2)edge[bend left]node[]{F}(s1);
			\draw(s1)edge[bend left]node[]{F}(s2);
			\draw(s2)edge[bend left]node[]{S}(s3);
			\draw[dashed](s3)edge[bend left]node[]{S}(s2);
			\draw(s0)edge[bend left]node[]{F}(s3);
			\draw[dashed](s3)edge[bend left]node[]{F}(s0);
		\end{tikzpicture}
    	
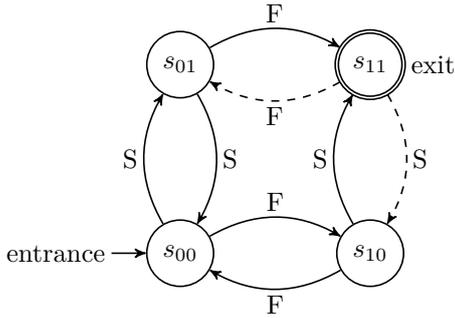
\captionof{figure}{A sample OQW with an absorbing boundary}
    	\label{fig:QW}
		\end{minipage}
		\begin{minipage}{\linewidth}
		\item[]
		\null\vfill\null\vfill
		As an absorbing boundary $s_{11}$ is introduced here,
		two transitions from it (in dashed line) do not exist anymore.
		Absorbing boundaries makes it impossible
		to characterise the system evolution by a closed system,
		i.e., using some Hermitian operator $\HH$ only,
		as usual in~\cite{SiP15}. 
		Fortunately, we could characterise it by an open system,
		i.e., using the linear operators $\LL_j$.
		\vfill\null\vfill \null
		\end{minipage}
		\end{itemize}
		\end{multicols}
\end{example}

\section{Signal Temporal Logic}\label{S4}
Here we recall the signal temporal logic (STL)~\cite{MaN04},
which is widely used to express real-time properties.
Using it we could specify richer properties of QCTMC
than linear temporal logic (LTL) and computation tree logic (CTL) in the time-bounded fragment.

\begin{definition}\label{def:syntax}
	The syntax of the STL formulas are defined as follows:
	\[
	\phi := \Phi \mid \neg\phi \mid \phi_1 \wedge \phi_2 \mid \phi_1 \ntl^\I \phi_2
	\]
	in which the atomic propositions $\Phi$,
	interpreted as \emph{signals},
	are of the form $p(\x) \in \II$
	where $p$ is a $\mathbb{Q}$-polynomial in $\x=(x_s)_{s\in S}$
	and $\II$ is a rational interval,
	and $\I$ is a finite time interval.
	Here $\ntl$ is called the until operator,
	and $\phi_1 \ntl^\I \phi_2$ is the until formula.
\end{definition}
	
\begin{definition}\label{def:semantics}
	The semantics of the STL formulas interpreted on a QCTMC $\QC$ in Definition~\ref{def:QCTMC2}
	are given by the satisfaction relation $\models$:
	\begin{align*}
		\rho(t) & \models \Phi
		&& \textup{if } p(\x) \in \II
		\textup{ holds with }x_s = \tr(\PP_s(\rho(t))), \\
		\rho(t) & \models \neg\phi
		&& \textup{if } \rho(t) \not\models \phi, \\
		\rho(t) & \models \phi_1 \wedge \phi_2
		&& \textup{if } \rho(t) \models \phi_1 \wedge \rho(t) \models \phi_2, \\
		\rho(t) & \models \phi_1 \ntl^\I \phi_2
		&& \textup{if there exists a real number }t' \in \I, \\
		&&& \textup{such that }\forall\,t_1 \in [t,t+t') \,:\, \rho(t_1) \models \phi_1
		\textup{ and }\rho(t+t') \models \phi_2,
	\end{align*}
	where $\PP_s$ is the projector $\op{s}{s} \otimes \id$ onto classical state $s$.
\end{definition}

From the semantics, we can see that
$\Phi_1 \ntl^{\I_1} (\Phi_2 \ntl^{\I_2} \Phi_3)$
and $(\Phi_1 \ntl^{\I_1} \Phi_2) \ntl^{\I_2} \Phi_3$ have different meanings:
The former formula requires there exist $t'\in \I_1$ and $t''\in \I_2$
such that $\forall\,t_1\in[t,t+t'): \rho(t_1) \models \Phi_1$,
$\forall\,t_2\in[t+t',t+t'+t''): \rho(t_2) \models \Phi_2$ and $\rho(t+t'+t'') \models \Phi_3$;
while the latter formula requires there exists a $t''\in \I_2$
such that $\rho(t+t'') \models \Phi_3$
and for each $t_1\in[t,t+t'')$, there exists a $t'\in \I_1$
such that $\forall\,t_1\in[t_1,t_1+t'): \rho(t_1) \models \Phi_1$
and $\rho(t_1+t') \models \Phi_2$.
We usually use the \emph{parse tree}
to clarify the structure of an STL formula $\phi$.

The logic is very generic.
STL has more expressive atomic propositions than LTL,
as $\true \equiv p(\x) \in (-\infty,+\infty)$ and $\false \equiv p(\x) \in \emptyset$.
CTL has a two-stage syntax consisting of state and path formulas.
Negation and conjunction are allowed in only state formulas, not path ones.
Whereas, STL allows negation and conjunction in any subformulas.
Besides the standard Boolean calculus,
we can easily obtain a few derivations:
$\Diamond^\I \phi \equiv  \true \ntl^\I \phi$,
$\Box^\I \phi \equiv \neg(\Diamond^\I \neg \phi)$,
and $\phi_1 \mathrm{R}^\I \phi_2 \equiv \neg (\neg\phi_1 \ntl^\I \neg\phi_2)$
where $\mathrm{R}$ is the release operator.
	
\begin{example}\label{ex2}
    Consider the open quantum walk $\QC_1$ described in Example~\ref{ex1}.
    It is not hard to get
    the probabilities of the particle staying respectively in $s_{00},s_{01},s_{10},s_{11}$ as follows:
	\begin{align*}
		x_{0,0} &= \tr(\PP_{s_{0,0}}(\rho(t)))
			= \tfrac{1}{2}\exp(-\tfrac{2+\sqrt{2}}{2}t)
			+ \tfrac{1}{2}\exp(-\tfrac{2-\sqrt{2}}{2}t), \\
		x_{0,1} &= \tr(\PP_{s_{0,1}}(\rho(t)))
			= -\tfrac{\sqrt{2}}{4}\exp(-\tfrac{2+\sqrt{2}}{2}t)
			+ \tfrac{\sqrt{2}}{4}\exp(-\tfrac{2-\sqrt{2}}{2}t) \\
			&\quad + \tfrac{1}{4}[\exp(-\tfrac{3}{2}t)-\exp(-\tfrac{1}{2}t)]\cos(\tfrac{1}{2}t) 
			+\tfrac{1}{4}[\exp(-\tfrac{3}{2}t)+\exp(-\tfrac{1}{2}t)]\sin(\tfrac{1}{2}t), \\
		x_{1,0} &= \tr(\PP_{s_{1,0}}(\rho(t)))
			= -\tfrac{\sqrt{2}}{4}\exp(-\tfrac{2+\sqrt{2}}{2}t)
			+\tfrac{\sqrt{2}}{4}\exp(-\tfrac{2-\sqrt{2}}{2}t) \\
			&\quad -\tfrac{1}{4}[\exp(-\tfrac{3}{2}t)-\exp(-\tfrac{1}{2}t)]\cos(\tfrac{1}{2}t)
			- \tfrac{1}{4}[\exp(-\tfrac{3}{2}t)+\exp(-\tfrac{1}{2}t)]\sin(\tfrac{1}{2}t), \\
		x_{1,1} &= \tr(\PP_{s_{1,1}}(\rho(t)))
    	    = 1	+\tfrac{-1+\sqrt{2}}{2}\exp(-\tfrac{2+\sqrt{2}}{2}t)
    	    -\tfrac{1+\sqrt{2}}{2}\exp(-\tfrac{2-\sqrt{2}}{2}t).
	\end{align*}
    We can see that
    the probability of exiting via $s_{11}$ would be one as $t$ approaches infinity.
    A further question is asking
    how fast the particle would reach the exit.
    This is actually a question about convergence performance.
    It is worth monitoring the critical moment
    at which the probability of exiting via $s_{11}$
    equals that of staying at transient positions $s_{01},s_{10}$.
    A requirement to be studied is like the following property.
	\begin{quote}
		\textbf{Property A:} At each moment $t\in[0,5]$,
		whenever the probability of staying at $s_{01}$ or $s_{10}$
		is greater than $\tfrac{1}{5}$,
		the probability of exiting via $s_{11}$ would overweight
		than that of staying at $s_{01}$ or $s_{10}$
		within the coming one unit of time.
	\end{quote}
	We formally specify it with the rich STL formula
	\[
		\phi_1 \equiv \Box^{\I_1}\, (\Phi_1 \rightarrow \Diamond^{\I_2} \Phi_2)
		\equiv \neg(\true\,\ntl^{\I_1}\, (\Phi_1 \wedge \neg(\true\,\ntl^{\I_2} \Phi_2))),
	\]
	where $\I_1=[0,5],\I_2=[0,1]$ are time intervals
	and $\Phi_1 \equiv x_{0,1}+x_{1,0}>\tfrac{1}{5},
	\Phi_2 \equiv x_{1,1} \ge x_{0,1}+x_{1,0}$ are atomic propositions,
	a.k.a. signals. \qed
\end{example}

\section{Solving Atomic Propositions}\label{S5}
As a basic step to decide the STL formula,
we need to solve the atomic proposition $\Phi$.
That is,
we will compute all solutions w.r.t. $t$, in which $\rho(t) \models \Phi$ holds.
We achieve it by a reduction to the real root isolation
for a class of real-valued functions, \emph{exponential polynomials}.
Although real roots of exponential polynomials have been studied
in many existing literature~\cite{AMW08,COW16,HLX+18},
the ones to be isolated in this paper involve the complicated complex exponents.
So we develop a state-of-the-art real root isolation for them,
whose completeness is established on Conjecture~\ref{Schanuel}.

Given an atomic proposition $\Phi\equiv p(\x) \in \II$
(assuming that $\II$ is bounded),
we would like to determine the algebraic structure of
\begin{equation}\label{eq:EP1}
	\varphi(t)=(p(\x(t))-\inf\II)(p(\x(t))-\sup\II),
\end{equation}
with which we will design an algorithm for solving $\Phi\equiv p(\x) \in \II$.
The structure of $\varphi(t)$ depends on that of $x_s(t)=\tr(\PP_s(\rho(t)))$.
We claim that
each entry of $\rho(t)$ is of the exponential polynomial form
\begin{equation}\label{eq:EP}
	\beta_1(t) \exp(\alpha_1 t) + \beta_2(t) \exp(\alpha_2 t) + \cdots
	+ \beta_m(t) \exp(\alpha_m t),
\end{equation}
where $\beta_1(t),\ldots,\beta_m(t)$ are nonzero $\mathbb{A}$-polynomials
and $\alpha_1,\ldots,\alpha_m$ are distinct algebraic numbers.
It follows the facts:
\begin{enumerate}
	\item The governing matrix $\M$
		in the state $\rho(t)=\vl(\exp(\M\cdot t)\cdot\lv(\rho(0)))$ of the QCTMC
		takes algebraic numbers  as entries. 
	\item The characteristic polynomial of $\M$ is an $\mathbb{A}$-polynomial.
	    The eigenvalues $\alpha_1,\ldots,\alpha_m$ of $\M$ are algebraic,
	    as they are roots of that $\mathbb{A}$-polynomial by Lemma~\ref{lem:closed}.
		Those eigenvalues make up all exponents in~\eqref{eq:EP}.
	\item The entries of the matrix exponential $\exp(\M\cdot t)$ are in the form~\eqref{eq:EP}.
\end{enumerate}
The same structure holds for $x_s(t)$,
as $x_s(t)=\tr(\PP_s(\rho(t)))$ is simply a sum of some entries of $\rho(t)$.
Furthermore, $\varphi(t)$ is also of the exponential polynomial form~\eqref{eq:EP},
since it is a $\mathbb{Q}$-polynomial in $\x(t)=(x_s(t))_{s \in S}$.
If $\II$ is unbounded from below (resp.~above),
the left (resp.~right) factor could be removed from~\eqref{eq:EP1} for further consideration.

Next, we will isolate all real roots $\lambda_1,\ldots,\lambda_n$ of $\varphi(t)$
in a bounded interval $\mathcal{B}$ (to be specified in the next section).
Before stating the core isolation algorithm---Algorithm~\ref{isol},
an overview of the isolation procedure is provided in Fig.~\ref{flow}.
The instances to be treated can be roughly divided into two classes: 
one is trivial that can be solved by the classical methods, e.g.,~\cite{CoL76},
for ordinary polynomials; 
the other is nontrivial that can be solved by Algorithm~\ref{isol}
but should meet three requirements of the input in Algorithm~\ref{isol}.
After the preprocesses \textbf{Basis Finding}, \textbf{Polynomialisation} and \textbf{Factoring},
the two classes of instances can be separated and solved by the corresponding methods.

\begin{figure}
	\includegraphics[scale=0.4]{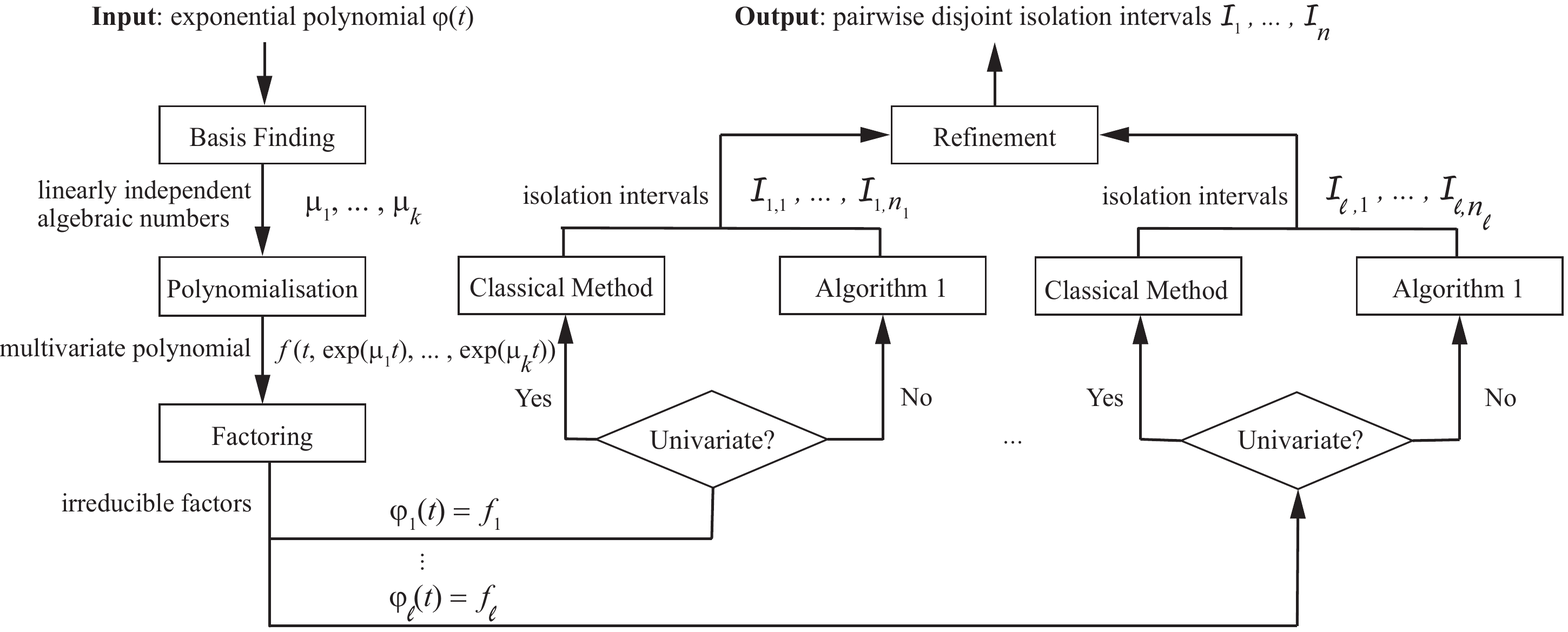}
	\caption{Flow chart of the whole isolation procedure}
	\label{flow}
\end{figure}

The technical details along the flow chart in Fig.~\ref{flow} are described below.
\begin{itemize}
	\item \textbf{Basis Finding}
	For a given set of algebraic numbers $\{\alpha_1,\ldots,\alpha_m\}$
	extracted from the exponents of the input exponential polynomial $\varphi(t)$,
	we can compute a simple extension
	$\mathbb{Q}(\mu):\mathbb{Q}=\mathbb{Q}(\alpha_1,\ldots,\alpha_m):\mathbb{Q}$
	by Lemma~\ref{lem:simple},
	such that each $\alpha_i=q_i(\mu)$ with $q_i\in \mathbb{Z}[x]$ and $\deg(q_i)<\deg(\mu)$;
	and further construct a $\mathbb{Q}$-linearly independent basis $\{\mu_1,\ldots,\mu_k\}$
	of those exponents $\{\alpha_1,\ldots,\alpha_m\}$ by~\cite[Section~3]{HLX+18},
	such that each $\alpha_i$ can be $\mathbb{Z}^+$-linearly expressed by $\{\mu_1,\ldots,\mu_k\}$.
	\item \textbf{Polynomialisation}
	Thus we can get a polynomial representation
	$f(t,\exp(\mu_1 t),\ldots,\linebreak[0]\exp(\mu_k t))$
	of $\varphi(t)$, where $f$ is a $(k+1)$-variate polynomial with algebraic coefficients.
	That is, $\varphi(t)$ is obtained by
	substituting $t,\exp(\mu_1 t),\ldots,\exp(\mu_k t)$ (as $k+1$ variables) into $f$.
	\item \textbf{Factoring}
	Factoring $\varphi(t)$ into irreducible factors $\varphi_i(t)$ ($1 \le i \le \ell$)
	corresponds to
	factoring the $(k+1)$-variate $\mathbb{A}$-polynomial $f$
	into irreducible factors $f_i$ ($1 \le i \le \ell$),
	which has been implemented in polynomial time, e.g.~\cite{Kal85}.
	\item \textbf{Univariate?}
	If the irreducible exponential polynomial $\varphi_i(t)$ corresponds
	a univariate polynomial $f_i$,
	isolating the real roots of $\varphi_i(t)$ can be treated by classical methods~\cite{CoL76};
	otherwise we will resort to Algorithm~\ref{isol},
	for which we can infer:
	\begin{enumerate}
		\item The exponential polynomial $\varphi(t)$ in the form~\eqref{eq:EP}
	    is plainly an \emph{analytic} function that is infinitely differentiable;
	    and it is a real-valued function,
	    or equivalently the imaginary part of $\varphi(t)$ is identically zero,
	    since each variable $x_s(t)$ is exactly the real-valued function $\tr(\PP_s \rho(t))$.
	    The same holds for any factor $\varphi_i(t)$ of $\varphi(t)$,
	    which ensures the first requirement of the input in Algorithm~\ref{isol}.
        \item By Theorem~\ref{Lindemann},
        thanks to the irreducibility of $\varphi_i(t)$,
	    we have $\varphi_i(\lambda) \ne 0$ holds for any $\lambda\in\mathbb{A}\setminus\{0\}$,
	    which ensures the second requirement of the input in Algorithm~\ref{isol}. 
	    \item By Conjucture~\ref{Schanuel} and Corollary~\ref{cor:common},
        each irreducible factor $\varphi_i(t)$ and its derivative $\varphi_i'(t)$
		are co-prime,
		and thus have no common real root except for $0$,
		which ensures the last requirement of the input in Algorithm~\ref{isol}.
	\end{enumerate}
	\item \textbf{Refinement}
		After performing Algorithm~\ref{isol}
		with each individual irreducible factor $\varphi_i(t)$ of $\varphi(t)$,
		we would obtain a list of disjoint isolation intervals $\I_{i,1},\ldots,\I_{i,n_i}$.
		The isolation intervals of different irreducible factors may be overlapping.
		However, by Corollary~\ref{cor:common} again,
		we have that
		each pair of co-prime factors of $\varphi(t)$ has no common real root except for $0$.
        So all these isolation intervals $\I_{i,1},\ldots,\I_{i,n_i}$ ($1 \le i \le \ell$)
		can be further refined to be pairwise disjoint.
		That would be the complete list of isolation intervals $\I_1,\ldots,\I_n$ for $\varphi(t)$,
		and thereby completes the whole isolation procedure.
\end{itemize}

\begin{algorithm}[ht!]
	\caption{\textsf{Real Root Isolation for a Real-valued Function}}\label{isol}
	\begin{algorithmic}[1]
		\item[] $$\{\I_1,\ldots,\I_n\} \Leftarrow {\sf Isolate}(\varphi,\I)$$
		\Require $\varphi(t)$ is a real-valued function
			defined on a rational interval $\I=[l,u]$,
			satisfying:
			\begin{enumerate}
			    \item $\varphi(t)$ is twice-differentiable,
				\item $\varphi(t)$ has no rational root in $\I$, and
				\item $\varphi(t)$ and $\varphi'(t)$ have no common real root in $\I$.
			\end{enumerate}
		\Ensure $\I_1,\ldots,\I_n$ are finitely many disjoint intervals,
			such that each contains exactly one real root of $\varphi$ in $\I$,
			and together contain all.
		\State compute an upper bound $M$ of $\{|\varphi'(t)|\,:\,t\in\I\}$;
		\State compute an upper bound $M'$ of $\{|\varphi''(t)|\,:\,t\in\I\}$;
		\State $i \gets 0$, $N \gets 2$ and $\delta \gets (u-l)/N$;
		\Comment{Here $N>1$ is a free parameter to indicate the number of subintervals to be split.
			We predefine it simply as $2$.}
		\While{$i \le N$}
		\If{$|\varphi(l+i\delta)|>M\delta$}\label{ln:unsat1}
			$i \gets i+1$;
		\Comment{$\varphi$ has no local real root}
		\ElsIf{$|\varphi(l+i\delta+\delta)|>M\delta$}\label{ln:unsat2}
			$i \gets i+2$;
		\Comment{$\varphi$ has no local real root}
		\ElsIf{$|\varphi'(l+i\delta)| \ge M'\delta$}\label{ln:dec1}
		\Comment{$\varphi$ is locally monotonic}
		\If{$\varphi(l+i\delta)\varphi(l+i\delta+\delta)<0$}\label{ln:sat1}
			\textbf{output} $(l+i\delta,l+i\delta+\delta)$;
		\EndIf
		\State $i \gets i+1$;
		\ElsIf{$|\varphi'(l+i\delta+\delta)| \ge M'\delta$}\label{ln:dec2}
		\Comment{$\varphi$ is locally monotonic}
		\If{$\varphi(l+i\delta)\varphi(l+i\delta+2\delta)<0$}\label{ln:sat2}
			\textbf{output} $(l+i\delta,l+i\delta+2\delta)$;
		\EndIf
		\State $i \gets i+2$;
		\Else
		\State\label{ln:rec} ${\sf Isolate}(\varphi,[l+i\delta,l+i\delta+\delta])$;
		\State $i \gets i+1$.
		\EndIf
		\EndWhile
	\end{algorithmic}
\end{algorithm}


\subparagraph*{Soundness of Algorithm~\ref{isol}}
We justify three groups of treatment in the loop body in turn.
\begin{enumerate}
	\item If the condition in Line~\ref{ln:unsat1} (resp.~\ref{ln:unsat2}) holds,
		$\varphi$ has no real root in the neighborhood
		centered at $l+i\delta$ (resp.~$l+(i+1)\delta$) with radius $\delta$.
		So we exclude the neighborhood.
	\item If the condition in Line~\ref{ln:dec1} (resp.~\ref{ln:dec2}) holds,
		$\varphi$ is monotonic in the neighborhood
		centered at $l+i\delta$ (resp.~$l+(i+1)\delta$) with radius $\delta$.
		Then, if the condition in Line~\ref{ln:sat1} (resp.~\ref{ln:sat2}) holds,
		i.e. $\varphi$ has different signs at endpoints of the neighborhood,
		the unique real root in the neighborhood exists and should be output;
		otherwise we just exclude the neighborhood.
	\item If it is not in the two decisive cases listed above,
		we perform Algorithm~\ref{isol} recursively
		in a subinterval $[l+i\delta,l+(i+1)\delta]$ of $[l,u]$. \qed
\end{enumerate}
	
	
\subparagraph*{Completeness of Algorithm~\ref{isol}}
The termination of Algorithm~\ref{isol} entails the completeness.
In other words, it suffices to show that
for any real-valued function $\varphi$ that satisfies the requirements,
Algorithm~\ref{isol} can always output all real roots of $\varphi$
within a finitely many times of recursion.
Since $\varphi$ and $\varphi'$ are real-valued functions defined
on the closed and bounded interval $\I$
and they have no common real root in $\I$,
there is a positive constant $\tau$,
such that either $|\varphi| \ge \tau$ or $|\varphi'| \ge \tau$ holds everywhere of $\I$.
Then, at any point $c$ in $\I$,
we can get a neighborhood with constant radius $rad:=\min(\tau/M,\tau/M')$,
in which either $\varphi$ has no real root or is monotonic.
So the two decisive cases must take place,
provided that the subinterval has length not greater than $rad$.
It implies that the recursion depth of Algorithm~\ref{isol} is bounded
by $\lceil \log_2(\|\I\|/rad) \rceil$, where $\|\I\| = \sup\I-\inf\I$.
Hence the termination is guaranteed. \qed

After obtaining all real roots $\lambda_1,\ldots,\lambda_n$ of $\varphi(t)$
in the bounded interval $\mathcal{B}$ by Algorithm~\ref{isol},
we have that
on each interval $\J_i$ ($0 \le i \le n$) of the $n+1$ intervals
in $\mathcal{B}\setminus\{\lambda_1,\ldots,\lambda_n\}$:
\begin{equation}\label{eq:interval}
	[\inf\mathcal{B},\lambda_1),(\lambda_1,\lambda_2),\ldots,
	(\lambda_{n-1},\lambda_n),(\lambda_n,\sup\mathcal{B}],
\end{equation}
$p(\x(t)) \in \II$ holds everywhere of $t\in \J_i$ or nowhere.
Finally, we can obtain the desired solution set $ \JJ$ (to be used in the next section)
of $p(\x(t)) \in \II$ by a finite union as follows
\begin{equation}
   \bigcup_{\substack{0 \le i \le n \\ p(\x(\J_i)) \subseteq \II}} \J_i  \cup
    \bigcup_{\substack{1 \le i \le n \\ p(\x(\lambda_i)) \in \II}} \{\lambda_i\}.
\end{equation}

\begin{example}\label{ex3}
Reconsider Example~\ref{ex2}. 
The exponential polynomial extracted from $\Phi_1$ is
\[
	\varphi_1(t) = x_{0,1}(t) + x_{1,0}(t) - \tfrac{1}{5}
	= -\tfrac{1}{5}	- \tfrac{\sqrt{2}}{2}\exp(-\tfrac{2+\sqrt{2}}{2}t)
	+ \tfrac{\sqrt{2}}{2}\exp(-\tfrac{2-\sqrt{2}}{2}t),
\]
and the exponential polynomial extracted from $\Phi_2$ is
\[
	\varphi_2(t) = x_{1,1}(t) - x_{0,1}(t) - x_{1,0}(t) 
	= 1+(\sqrt{2}-\tfrac{1}{2})\exp(-\tfrac{2+\sqrt{2}}{2}t)
	- (\sqrt{2}+\tfrac{1}{2})\exp(-\tfrac{2-\sqrt{2}}{2}t).
\] 
To solve $\Phi_1$ and $\Phi_2$ in a bounded interval, say $\mathcal{B}=[0,6]$,
we need to determine the real roots of $\varphi_1$ and $\varphi_2$.
Since both $\varphi_1$ and $\varphi_2$ are irreducible
and their corresponding polynomial representations are bivariate,
they have no rational root, repeated real root, nor common real root. 
After invoking Algorithm~\ref{isol} on $\varphi_1$ with $\mathcal{B}$,
we obtain two isolation intervals $[0,\tfrac{25}{64}]$
(containing real root $\lambda_1 \approx 0.352097$)
and $[\tfrac{275}{64},\tfrac{575}{128}]$ (containing $\lambda_2 \approx 4.49181$),
which could be easily refined up to any precision.
For $\varphi_2$,
as $\varphi_2(0) = 0$, 
we make a slight shift on the left endpoint of $\mathcal{B}$ in the consideration,
e.g., $[\tfrac{1}{1000},6]$.
After invoking Algorithm~\ref{isol} on $\varphi_2$ with $[\tfrac{1}{1000},6]$,
we could get the unique isolation intervals	$[\tfrac{125059}{64000},\tfrac{275117}{128000}]$,
which contains the real root $\lambda_3 \approx 2.14897$.
The three isolation intervals are pairwise disjoint.

Finally, we have that
the solution set of $\Phi_1$ is $(\lambda_1,\lambda_2)$ in $\mathcal{B}$,
and the solution set of $\Phi_2$ is $\{0\} \cup [\lambda_3,6]$. \qed
\end{example}

Under Conjecture~\ref{Schanuel}, we get the following computability result
while the complexity is still an open problem as left in existing literature~\cite{COW16,HLX+18}.
\begin{lemma}
	The atomic propositions in STL are solvable over QCTMCs.
\end{lemma}

In fact, Conjecture~\ref{Schanuel} is a powerful tool
to treat roots of the general exponential polynomial.
For some special subclasses of exponential polynomials,
there are solid theorems to treat them:
one is Theorem~\ref{Lindemann} that has been employed~\cite{AMW08}
for the exponential polynomials in the ring $\mathbb{Q}[t,\exp(t)]$,
the other is the Gelfond--Schneider theorem employed~\cite[Subsection~4.1]{HLX+18}
for the exponential polynomials in $\mathbb{Q}[\exp(\mu_1 t),\exp(\mu_2 t)]$
where $\mu_1$ and $\mu_2$ are
\emph{two} $\mathbb{Q}$-linear independent \emph{real} algebraic numbers.
In Example~\ref{ex3}, the Gelfond--Schneider theorem is sufficient to
treat roots of the exponential polynomials $\varphi_1$ and $\varphi_2$,
thus the termination is guaranteed.
Algorithm~\ref{isol} can isolate roots of elements in
$\mathbb{Q}[t,\exp(\mu_1 t),\ldots,\exp(\mu_k t)]$
for \emph{arbitrarily many} $\mathbb{Q}$-linear independent
\emph{complex} algebraic numbers $\mu_1,\ldots,\mu_k$.
Additionally, Conjecture~\ref{Schanuel} has been employed
to isolate simple roots of more expressible functions
than exponential polynomials~\cite{Str08,Str09},
but fails to find repeated roots.
Hence, this paper makes a trade-off
between the expressibility of functions and the completeness of methodologies.

\section{Checking STL Formulas}\label{S6}
In the previous section, we have solved atomic propositions.
Now we consider the general STL formula $\phi$.
For a given formula $\phi$,
we compute the so-called \emph{post-monitoring period} $\mnt(\phi)$,
independent on the initial time $t_0$,
such that the truth of $\rho(t_0)\models\phi$ could be affected
by those $\rho(t)$ with $t \in [t_0,t_0+\mnt(\phi)]$.
Then we decide $\phi$ with a bottom-up fashion.
The complexity turns out to be linear in the size $\|\phi\|$ of the input STL formula $\phi$,
which is defined as the number of logic connectives in $\phi$ as standard.

Given an STL formula $\phi$,
we need to post-monitor a time period to decide the truth of $\rho(t_0)\models\phi$
at an initial time $t_0$,
especially for the until formula $\phi_1 \ntl^\I \phi_2$.
For example, according to the semantics of STL,
to decide $\rho(t_0)\models \phi_2$
with $\phi_2 \equiv \Phi_1 \ntl^{\I_1} (\Phi_2 \ntl^{\I_2} \Phi_3)$,
we have to monitor the states $\rho(t)$ from time $t_0$ to $t_0+\sup\I_1+\sup\I_2$.
It inspires us
to calculate the post-monitoring period $\mnt(\phi)$ as
\begin{equation}\label{mnt}
	\mnt(\phi)=\begin{cases}
		0 & \textup{if }\phi=\Phi, \\
		\mnt(\phi_1) & \textup{if }\phi=\neg\phi_1, \\
		\max(\mnt(\phi_1),\mnt(\phi_2))	& \textup{if }\phi=\phi_1 \wedge \phi_2, \\
		\sup\I+\max(\mnt(\phi_1),\mnt(\phi_2)) & \textup{if }\phi=\phi_1 \ntl^\I \phi_2.
	\end{cases}
\end{equation}

\begin{lemma}\label{lem:mnt}
	Given an STL formula $\phi$,
	the satisfaction $\rho(t_0) \models \phi$ is entirely determined by
	the states $\rho(t)$ with $t_0 \le t \le t_0+\mnt(\phi)$.
\end{lemma}
\begin{proof}
	We discuss it upon the syntactical structure of the STL formula $\phi$.
	
	For the atomic proposition $\Phi$,
	$\rho(t_0) \models \Phi$ is plainly determined by $\rho(t_0)$.
	    	
	For the negation $\neg\phi_1$,
	if $\rho(t_0) \models \phi_1$
	is determined by $\rho(t)$ with $t_0 \le t \le t_0+\mnt(\phi_1)$,
	so is $\rho(t_0) \models \neg\phi_1$.
			
	For the conjunction $\phi_1 \wedge \phi_2$,
	if $\rho(t_0) \models \phi_1$ (resp.~$\rho(t_0) \models \phi_2$)
	is determined by $\rho(t)$ with $t_0 \le t \le t_0+\mnt(\phi_1)$
	(resp.~$t_0 \le t \le t_0+\mnt(\phi_2)$),
	$\rho(t_0) \models \phi_1 \wedge \phi_2$ is determined by the union of those states,
	i.e. $\rho(t)$ with $t_0 \le t \le t_0+\max(\mnt(\phi_1),\mnt(\phi_2))$.
			
	For the until formula $\phi_1 \ntl^\I \phi_2$,
	if $\rho(t_0) \models \phi_1$ (resp.~$\rho(t_0) \models \phi_2$)
	is determined by $\rho(t)$ with $t_0 \le t \le t_0+\mnt(\phi_1)$
	(resp.~$t_0 \le t \le t_0+\mnt(\phi_2)$),
	$\rho(t_0) \models \phi_1 \ntl^\I \phi_2$ is determined by
	$\rho(t)$ with $t_0 \le t \le t_0+\sup\I+\max(\mnt(\phi_1),\mnt(\phi_2))$,
	where $\sup\I$ is caused by the admissible transition at the latest time in the until formula,
	since then we have to determine $\rho(t) \models \phi_1$ from time $t_0$ to $t_0+\sup\I$
	and determine $\rho(t_0+\sup\I) \models \phi_2$.
\end{proof}
	
To decide $\rho(0) \models \phi$,
our method is based on the parse tree $\mathcal{T}$ of $\phi$ as follows.

Basically, for each leaf of $\mathcal{T}$
that represents an atomic proposition $\Phi$,
we compute the solution set $\JJ$ (possibly a union of maximal solution intervals $\J$)
of $\Phi$ within the monitoring interval $\mathcal{B}:=[0,\mnt(\phi)]$
by Algorithm~\ref{isol}.

Inductively, for each intermediate node of $\mathcal{T}$
that represents the subformula $\psi$ of $\phi$,
we tackle it into three classes.
\begin{itemize}
	\item If $\psi=\neg \phi_1$,
		supposing that $\JJ_1$ is the solution set of $\phi_1$,
		the solution set $\JJ'$ of $\psi$
		is $\mathcal{B} \setminus \JJ_1$.
	\item If $\psi=\phi_1 \wedge \phi_2$,
		supposing that $\JJ_1$ (resp.~$\JJ_2$) is the solution set of
		$\phi_1$ (resp.~$\phi_2$),
		the solution set $\JJ'$ of $\psi$ is $\JJ_1 \cap \JJ_2$.
	\item If $\psi=\phi_1 \ntl^\I \phi_2$,
		supposing that $\JJ_1$ (resp.~$\JJ_2$) is the solution set of
		$\phi_1$ (resp.~$\phi_2$),
		the solution set $\JJ'$ of $\psi$ is
		\[
		\{t : (t'\in\I) \wedge ([t,t+t') \subseteq \JJ_1 )\wedge (t+t'\in\JJ_2)\}.
		\]
		directly from the semantics of the until formula $\phi_1 \ntl^\I \phi_2$
		in Definition~\ref{def:semantics}, as
		\begin{itemize}
			\item $[t,t+t') \subseteq \JJ_1$
			if and only if $\forall\,t_1 \in [t,t+t') \,:\, \rho(t_1) \models \phi_1$, and 
			\item $t+t'\in\JJ_2$ if and only if $\rho(t+t') \models \phi_2$.
		\end{itemize}
    \end{itemize}
Note that the inductive steps of the above procedure
do not generally produce all solutions of the subformula $\psi$ in $\mathcal{B}$.
Since by Lemma~\ref{lem:mnt}
$\rho(t_0) \models \psi$ is entirely determined
by $\rho(t)$ with $t_0 \le t \le t_0+\mnt(\psi)$,
the resulting solution set $\JJ'$ contains all solutions of $\psi$ in $[0,\mnt(\phi)-\mnt(\psi)]$
and possibly misses some solutions in the right subinterval
$(\mnt(\phi)-\mnt(\psi),\mnt(\phi)]$.
Anyway, the subinterval $[0,\mnt(\phi)-\mnt(\psi)]$ has the left-closed endpoint $0$,
which suffices to decide $\rho(0) \models \phi$.

With a bottom-up fashion,
we could eventually get the solution set $\JJ$ of $\phi$,
by which $\rho(0) \models \phi$ can be decided to be true if and only if $0 \in \JJ$.
Overall, the procedure costs $\|\phi\|$ times of the interval operations
and at most $\|\phi\|$ times of calling Algorithm~\ref{isol}
for getting the solution set of an atomic proposition.
That is, the query complexity of model checking STL formulas is linear in $\|\phi\|$
by calling Algorithm~\ref{isol}.
The query complexity addresses the issue of the number of calls to a black box routine
with unknown complexities and is commonly used in quantum computing,
where the routine is usually called an oracle.
For example,
oracles can be a procedure of encoding an entry of a matrix into a quantum state~\cite{HHL09}
or a quantum circuit of preparing a specific quantum state~\cite{GWY21}.

\begin{example}
	For the STL formula
	$\phi_1 \equiv \neg(\true\,\ntl^{\I_1}\,(\Phi_1 \wedge \neg(\true\,\ntl^{\I_2} \Phi_2)))$,
	the post-monitoring period $\mnt(\phi_1)$ is $\sup\I_1+\sup\I_2=6$ by Eq.~\eqref{mnt},
	implying $\mathcal{B}=[0,6]$ as used in Example~\ref{ex3}.
	We construct the parse tree of $\phi_1$
	in Figure~\ref{fig:parse}.
	Based on it,
	we could calculate the solution sets of all nodes in a bottom-up fashion:
	\begin{itemize}
		\item Basically, the solution set for $\Phi_1$ is $(\lambda_1,\lambda_2)$
		where $\lambda_1 \approx 0.352097$ and $\lambda_2 \approx 4.49181$,
		the solution set for $\Phi_2$ is $\{0\}\cup [\lambda_3,6]$ where $\lambda_3 \approx 2.14897$,
		and the solution set for $\true$ is plainly $[0,6]$.
		The post-monitoring periods of the three leaves are $0$. 
		Since $\mathcal{B}$ is $[0,6]$, 
		we have got all the solutions in $[0,6]$.
		\item The solution set for $\true\,\ntl^{\I_2} \Phi_2$ is $[\lambda_3-1,6]$,
		and that for the negation is $[0,\lambda_3-1)$.
		The post-monitoring periods of the two nodes are $1$,
		thus we have got all solutions in $[0,5]$.
		\item The solution set for $\Phi_1 \wedge \neg(\true\,\ntl^{\I_2} \Phi_2)$
		is $(\lambda_1,\lambda_3-1)$.
		Its post-monitoring period is $1$, 
		thus we have got all solutions in $[0,5]$.
		\item Finally, the solution set for
		$\true\,\ntl^{\I_1}\, (\Phi_1 \wedge \neg(\true\,\ntl^{\I_2} \Phi_2))$ is $[0,\lambda_3-1)$,
		and that for the negation (the root, representing the whole STL formula $\phi_1$)
		is $[\lambda_3-1,6]$.
		Their post-monitoring periods are $6$,
		thus we have got the solution in $[0,0]$.
	\end{itemize}
    Since $0\notin [\lambda_3-1,6]$, 
    we can decide $\rho(0) \models \phi_1$ to be false.
    Hence the particle walking along Figure~\ref{fig:QW} does not satisfy
    the desired convergence performance---Property A. \qed
		
\begin{figure}[ht]
	\begin{tikzpicture}[->,>=stealth',auto,node distance=1.8cm,semithick,inner sep=2pt]
		\node[state,rectangle](s1){$\neg$};
		\node[state,rectangle](s2)[right of=s1]{$\ntl^{\I_1}$};
		\node[state,rectangle](s3)[right of=s2]{$\wedge$};
		\node[state,rectangle](s4)[below of=s3]{$\true$};
		\node[state,rectangle](s5)[right of=s3]{$\neg$};
		\node[state,rectangle](s6)[below of=s5]{$\Phi_1$};
		\node[state,rectangle](s7)[right of=s5]{$\ntl^{\I_2}$};
		\node[state,rectangle](s8)[right of=s7]{$\Phi_2$};
		\node[state,rectangle](s9)[below of=s8]{$\true$};
	
		\draw[->](s1)edge[]node{}(s2);
		\draw[->](s2)edge[]node{}(s3);
		\draw[->](s2)edge[]node{}(s4);
		\draw[->](s3)edge[]node{}(s5);
		\draw[->](s3)edge[]node{}(s6);
		\draw[->](s5)edge[]node{}(s7);
		\draw[->](s7)edge[]node{}(s8);
		\draw[->](s7)edge[]node{}(s9);
	\end{tikzpicture}
	\caption{Parse tree of the STL formula $\phi_1$}\label{fig:parse}
\end{figure}
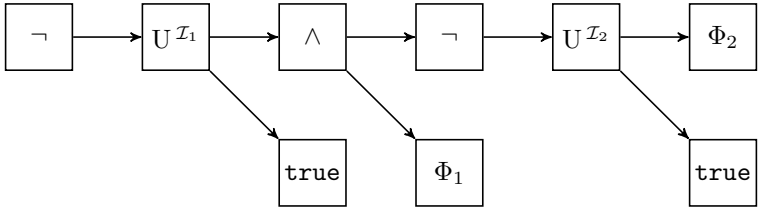

\end{example}

Finally, under Conjecture~\ref{Schanuel}, we obtain the main result:
\begin{theorem}
	The STL formulas are decidable over QCTMCs.
\end{theorem}

\section{Concluding Remarks}\label{S7}
In this paper, we introduced the model of QCTMC that extends CTMC,
and established the decidability of the STL formulas over it.
To this goal, we firstly solved the atomic propositions in STL
by real root isolation of a wide class of exponential polynomials,
whose completeness was based on Schanuel's conjecture.
Then we decided the general STL formula using interval operations with a bottom-up fashion,
whose (query) complexity turns out to be linear in the size of the input formula
by calling the developed state-of-art real root isolation routine.
We demonstrated our method by a running example of an open quantum walk.

For future work,
we would like to explore the following aspects:
\begin{itemize}
	\item how to apply the proposed method to verify non-Markov models in the real world~\cite{Pel14};
    \item how to design an efficient numerical approximation of the exact method in this paper;
    \item and checking other formal logics, e.g.~\cite{ASS+96}, over the QCTMC.
\end{itemize}

\bibliography{ref}
	
		
\end{document}